\documentclass[11pt]{article}

\usepackage{times}
\usepackage{fullpage}
\usepackage{amsfonts,amssymb,amsmath}
\usepackage{latexsym}
\usepackage{url}
\usepackage{verbatim}
\usepackage[pdftex]{color}
\usepackage[breaklinks]{hyperref}

\newcommand{\R}{\mathbb{R}}
\newcommand{\eps}{\varepsilon}
\newcommand{\norm}[1]{\mbox{$\parallel{#1}\parallel$}}

\newcommand{\floor}[1]{\lfloor #1 \rfloor}
\newcommand{\ceil}[1]{\left\lceil #1 \right\rceil}

\newcommand{\braket}[2]{\langle #1|#2 \rangle}
\newcommand{\ketbra}[2]{| #1\rangle \langle #2|}
\newcommand{\ket}[1]{| #1 \rangle}
\newcommand{\bra}[1]{\langle #1 |}

\def\01{\{0,1\}}

\newcommand{\QE}{\mathrm{QE}}
\newcommand{\QCE}{\mathrm{QCE}}
\newcommand{\RE}{\mathrm{RE}}
\newcommand{\Tr}{\mathrm{Tr}}
\newcommand{\rk}{\mathrm{rk}}
\newcommand{\polylog}{\mathrm{polylog}}
\newcommand{\psdrk}{\mathrm{rk_{psd}}}

\newcommand{\diag}{\mathrm{diag}}
\newcommand{\conv}{\mathrm{conv}}
\newcommand{\E}{\mathbb{E}}

\newcommand{\ldeg}{\mathrm{ldeg}}

\newcommand{\ignore}[1]{}

\newtheorem{theorem}{Theorem}
\newtheorem{lemma}[theorem]{Lemma}

\newtheorem{definition}[theorem]{Definition}

\newcommand{\thmref}[1]{\hyperref[#1]{{Theorem~\ref*{#1}}}}
\newcommand{\lemref}[1]{\hyperref[#1]{{Lemma~\ref*{#1}}}}
\newcommand{\corref}[1]{\hyperref[#1]{{Corollary~\ref*{#1}}}}
\newcommand{\eqnref}[1]{\hyperref[#1]{{Equation~(\ref*{#1})}}}
\newcommand{\claimref}[1]{\hyperref[#1]{{Claim~\ref*{#1}}}}
\newcommand{\remarkref}[1]{\hyperref[#1]{{Remark~\ref*{#1}}}}
\newcommand{\propref}[1]{\hyperref[#1]{{Proposition~\ref*{#1}}}}
\newcommand{\factref}[1]{\hyperref[#1]{{Fact~\ref*{#1}}}}
\newcommand{\defref}[1]{\hyperref[#1]{{Definition~\ref*{#1}}}}
\newcommand{\exampleref}[1]{\hyperref[#1]{{Example~\ref*{#1}}}}
\newcommand{\hypref}[1]{\hyperref[#1]{{Hypothesis~\ref*{#1}}}}
\newcommand{\secref}[1]{\hyperref[#1]{{Section~\ref*{#1}}}}
\newcommand{\chapref}[1]{\hyperref[#1]{{Chapter~\ref*{#1}}}}
\newcommand{\apref}[1]{\hyperref[#1]{{Appendix~\ref*{#1}}}}

\newenvironment{proof}[1][Proof: ]
{\noindent {\bf #1}}
{{\hfill $\Box$}\\
 \smallskip}

\begin{document}

\title{Query complexity in expectation}
\author{Jedrzej Kaniewski\thanks{Centre for Quantum Technologies, National University of Singapore and QuTech, Delft University of Technology, the Netherlands. \url{j.kaniewski@nus.edu.sg}.} 
\and 
Troy Lee\thanks{School of Mathematics and Physical Sciences, Nanyang Technological University and Centre for Quantum Technologies, Singapore. \url{troyjlee@gmail.com}.  This material is based on research supported in part by the Singapore National Research Foundation under NRF RF Award No. NRF-NRFF2013-13.}  
 \and 
Ronald de Wolf\thanks{CWI and University of Amsterdam, the Netherlands. \url{rdewolf@cwi.nl}. Partially supported by a Vidi grant from the Netherlands Organization for Scientific Research (NWO) which ended in 2013, by ERC Consolidator Grant QPROGRESS, and by the European Commission IST STREP project Quantum Algorithms (QALGO) 600700.}}
\date{}
\maketitle
\thispagestyle{empty}

\begin{abstract}
We study the query complexity of computing a function $f:\01^n\to\R_+$ \emph{in expectation}. This requires the algorithm on input $x$ to output a nonnegative random variable whose expectation equals~$f(x)$, using as few queries to the input~$x$ as possible. We exactly characterize both the randomized and the quantum query complexity by two polynomial degrees, the nonnegative literal degree and the sum-of-squares degree, respectively. We observe that the quantum complexity can be unboundedly smaller than the classical complexity for some functions, but can be at most polynomially smaller for functions with range $\01$.

These query complexities relate to (and are motivated by) the extension complexity of polytopes.  The \emph{linear} extension complexity of a polytope is characterized by the randomized \emph{communication} complexity of computing its slack matrix in expectation, and the \emph{semidefinite} (psd) extension complexity is characterized by the analogous quantum model.  Since query complexity can be used to upper bound communication complexity of related functions, we can derive some upper bounds on psd extension complexity by constructing efficient quantum query algorithms. As an example we give an exponentially-close entrywise  approximation of the slack matrix of the perfect matching polytope with psd-rank only $2^{n^{1/2+\eps}}$.  Finally, we show there is a precise sense in which randomized/quantum query complexity in expectation corresponds to the Sherali-Adams and Lasserre hierarchies, respectively.
\end{abstract}

\newpage
\setcounter{page}{1}

\section{Introduction}

\subsection{Computing functions in expectation}

We study the complexity of computing a function $f:\01^n \to \R_+$ \emph{in expectation}.  In this setting, on input~$x$ we want our algorithm to output a nonnegative
real number whose expectation (over the algorithm's internal randomness) exactly equals~$f(x)$.  Getting the expectation right is an easier task than computing the function value $f(x)$ itself, and suffices in some applications. For example, suppose we want to approximate the value $F(x)=\sum_{i=1}^m f_i(x)$ that depends on $x\in\01^n$. Then we can just compute each $f_i(x)$ \emph{in expectation} and output the sum of the results. By linearity of expectation, the output will have expectation~$F(x)$, and it will be tightly concentrated around its expectation if the random variables are not too wild (so the Central Limit Theorem applies). It is not necessary to compute or even approximate any of the values $f_i(x)$ themselves for this. This illustrates that computing functions in expectation is an interesting model in its own right. Additionally, it is motivated by connections with the \emph{extension complexity} of polytopes that are used in combinatorial optimization (roughly: the minimal size of linear or semidefinite programs for optimizing over such a polytope), as we describe below in Section~\ref{ssecextensioncomplexity}. 

The complexity of computing~$f$ can be measured in different ways, and here we will focus on \emph{query} complexity. We measure the complexity of computing a function in expectation by the (worst-case) number of queries to the input $x\in\01^n$ that the best algorithm uses. We study both \emph{randomized} and \emph{quantum} versions of this model and show that both of these query complexities can be exactly characterized by natural notions of polynomial degree. In Section~\ref{sec:re} we show that the randomized query complexity of computing $f$ in expectation equals the ``nonnegative literal degree'' of~$f$, which is the minimal~$d$ such that $f$ can be written as a nonnegative linear combination of products of up to $d$ variables or negations of variables.  In Section~\ref{sec:qe} we show that the quantum complexity equals the ``sum-of-squares degree'', which is the minimal~$d$ such that there exist polynomials $p_i$ of degree at most $d$ satisfying $f(x)=\sum_i p_i(x)^2$ for all $x\in\01^n$. 

In Section~\ref{secqcgap} we observe that quantum and classical query complexities (equivalently: the above two types of polynomial degree) can be arbitrarily far apart. For example, the function $f(x)=(\sum_{i=1}^n x_i -1)^2$ is the square of a degree-1 polynomial and hence can be computed in expectation with only 1 quantum query, while randomized algorithms need $n$ queries to get this expectation right. In contrast, we also show that for functions with range $\01$, the quantum-classical gap cannot be very large: at most cubic. 

Lower bounds on the quantum query complexity can be obtained from lower bounding the sum-of-squares degree of the function at hand, which is often non-trivial.  In Section~\ref{ssecsoslowerbound}, using techniques from approximation theory, we prove that the function $f(x)=(\sum_{i=1}^n x_i -1)(\sum_{i=1}^n x_i -2)$ has sum-of-squares degree $\Omega(\sqrt{n})$.  Hence quantum algorithms require $\Omega(\sqrt{n})$ queries to compute this function in expectation.

\subsection{Motivation: linear and semidefinite extension complexity}\label{ssecextensioncomplexity}

Our main motivation for studying query complexity in expectation comes from combinatorial optimization, in particular from linear and semidefinite programs. Many optimization problems can be formulated as maximizing or minimizing a linear function over a polytope. For example, in the Traveling Salesman Problem on $n$-vertex undirected graphs, one wants to minimize a linear function (the length of the tour) over the polytope $P\subseteq\R^{n\choose 2}$ that is the convex hull of all Hamiltonian cycles in the complete $n$-vertex graph~$K_n$.  If this polytope could somehow be represented as the feasible region of a small linear or semidefinite program, then we could efficiently solve the problem using the ellipsoid or interior-point methods.

Informally, the {\bf linear extension complexity} of a polytope~$P\subseteq \R^d$ is the minimum number of linear inequalities (over the $d$ variables of $P$ as well as possibly auxiliary variables) whose feasible region projects down to~$P$. If the linear extension complexity is small, there is a small linear program to optimize over~$P$. 

Motivated by erroneous claims~\cite{swart:tsp} that the TSP polytope had polynomial linear extension complexity (implying P $=$ NP), Yannakakis~\cite{Yannakakis91} showed that ``symmetric'' linear extensions of the Traveling Salesman Polytope need $2^{\Omega(n)}$ linear inequalities. He showed the same for the perfect matching polytope (which is spanned by all perfect matchings in $K_n$), despite the fact that finding a maximum matching can be done efficiently! For a long time, generalizing these lower bounds to arbitrary (possibly non-symmetric) linear extensions was an open question.  However, recently Fiorini et al.~\cite{FioriniMassarPokuttaTiwaryDewolf2012} proved a $2^{\Omega(n^{1/2})}$ lower bound on the linear extension complexity of the TSP polytope.  Subsequently Rothvo\ss~\cite{rothvoss:matching} proved a $2^{\Omega(n)}$ lower bound for the perfect matching polytope, which via a reduction implies the same bound for TSP. Chan et al.~\cite{clrs:csplp} obtained lower bounds on linear extension complexity for constraint satisfaction problems via a different route: roughly put, they showed that arbitrary linear extensions are not much more powerful than the specific linear extensions produced by the ``Sherali-Adams Hierarchy''; hence they could obtain lower bounds on linear extension complexity from known bounds on the Sherali-Adams hierarchy.

The {\bf positive semidefinite (psd) extension complexity} of polytope~$P$ is similar, but replaces the linear programs by potentially more powerful semidefinite programs.  The complexity is now the minimal dimension of a semidefinite program whose feasible region projects down to~$P$.  In contrast to the case of linear extension complexity, very few lower bounds on psd extension complexity are known. Until recently, there were only a few lower bounds for ``symmetric'' psd extensions~\cite{lrst:symmsdp,fsp:equivariant}.  However, in a \emph{very} recent breakthrough, Lee et al.~\cite{lrs:cspsdp} generalized the approach of~\cite{clrs:csplp} to show that arbitrary psd extensions are not much more powerful than the specific psd extensions produced by the ``Lasserre Hierarchy''. In particular they showed that the TSP polytope has psd extension complexity $2^{\Omega(n^{1/13})}$. 

Surprisingly, there is a very close connection between these extension complexities and the model of computing functions in expectation, albeit for the \emph{communication complexity} of computing a 2-input function.  More precisely, suppose Alice receives input~$x$, Bob receives input~$y$, and they want to compute some function $g(x,y)$ (which may also be viewed as a matrix).  In the usual setting of communication complexity~\cite{KushilevitzNisan97}, one of the parties (let's say Bob) has to output this value~$g(x,y)$ exactly, either with probability~1 or with high probability.  However, we may also consider how much communication they need to compute $g(x,y)$ \emph{in expectation}, i.e., now Bob needs to output a nonnegative random variable whose expected value equals~$g(x,y)$. Faenza et al.~\cite{ffgt:extandcc} showed that the logarithm of the linear extension complexity of a polytope~$P$ equals the randomized communication complexity of computing (in expectation) a matrix associated with $P$, known as the \emph{slack matrix}.  Lifting this result to the quantum/psd case, Fiorini et al.~\cite{FioriniMassarPokuttaTiwaryDewolf2012} showed that the logarithm of the \emph{psd} extension complexity equals the one-way \emph{quantum} communication complexity of computing the slack matrix of $P$ in expectation; in this model Alice sends a single quantum message to Bob. These connections show that studying (linear and psd) extension complexity of a polytope~$P$ is \emph{equivalent} to studying (randomized and one-way quantum) communication complexity in expectation, of the slack matrix of~$P$.

How do our results on the \emph{query} complexity of computing a function in expectation impact this \emph{communication} complexity?  Many functions of interest in communication complexity are of the form $g(x,y)=f(x\wedge y)$ for some Boolean function $f:\01^n \to \01$, where the AND-connective is applied bitwise.  Functions of this form also arise as (submatrices of) slack matrices of interesting polytopes, for example the correlation polytope.  Quite generally across the usual models of worst-case complexity, be it deterministic, randomized, or quantum, upper bounds on the \emph{query complexity} of~$f$ imply upper bounds on the \emph{communication complexity} of~$g$. In Section~\ref{sec:comm-query} we show that this also holds for the randomized and quantum models of computing a function in expectation.  As this leads to multi-round communication protocols, we also show that the one-way quantum communication complexity of computing a function in expectation equals the two-way complexity.

In Section~\ref{ssec:matchingpolytope} we give an application of the connection between query algorithms and communication complexity (equivalently, \emph{psd rank}), by deriving an exponentially-close entrywise approximation of the slack matrix $S$ of the perfect matching polytope with psd rank $2^{n^{1/2+\eps}}$. This psd rank is surprisingly low in view of the fact that Rothvo\ss~\cite{rothvoss:matching} showed that the nonnegative rank of $S$ is $2^{\Omega(n)}$, and Braun and Pokutta~\cite{braun&pokutta:matchrelax} showed that any $\tilde{S}$ that is $O(1/n)$-close to~$S$ still needs nonnegative rank $2^{\Omega(n)}$. 


Communication protocols derived from query algorithms have a specific structure.  In spirit, this is somewhat 
similar to looking at linear/psd extensions derived from hierarchies of specific linear or semidefinite programs 
like the Sherali-Adams and Lasserre hierarchies.  We show that these two relaxations actually correspond in 
a precise sense: just as the linear and psd extension complexities are characterized by models of communication complexity in expectation, the Sherali-Adams and Lasserre hierarchies are characterized by randomized and quantum models of query complexity in expectation, respectively.  This connection, described in Section~\ref{sec:hierarchies}, follows from known characterizations of these hierarchies in terms of notions of polynomial degrees which exactly correspond to the polynomial degrees we consider here.

\section{Preliminaries}

\subsection{Polytopes and extension complexity}\label{sec:defs}

A polytope $P\subseteq\R^d$ has both an \emph{inner description} as the convex hull of a set~$V\subseteq\R^d$ of points, $P=\conv(V)$; and an \emph{outer description} as the intersection of halfspaces, $P=\{x\in\R^d : Ax \le b\}$.  A \emph{slack matrix} integrates information from these two descriptions:

\begin{definition}
Let $P=\conv(V)=\{x : Ax \le b\}$ be a polytope.  The slack matrix $M$ of $P$ has columns labeled by $v \in V$ and rows labeled by constraints $A_i x \le b_i$, with entries~$M(i,v)=b_i-A_iv$.  
\end{definition}

\begin{definition}
Let $M$ be a nonnegative matrix.  A nonnegative factorization of $M$ of size $d$ consists of two sets of $d$-dimensional nonnegative vectors $\{a_x\},\{b_y\}$ such that $M(x,y)=a_x^T b_y$ for all $x,y$.  The nonnegative rank of $M$, denoted $\rk_+(M)$, is the minimal size among all nonnegative factorizations of $M$.
Equivalently, it is the minimum number of nonnegative rank-one matrices whose sum is $M$.
\end{definition}

\begin{definition}
Let $M$ be a nonnegative matrix.  A psd factorization of $M$ of size $d$ consists of two sets of $d$-by-$d$ psd matrices $\{A_x\},\{B_y\}$ such that $M(x,y)=\Tr(A_x B_y)$ for all $x,y$.  The psd rank of $M$, denoted $\psdrk(M)$, is the minimal size among all psd factorizations of $M$.
\end{definition}
Note that a nonnegative factorization is a psd factorization where the matrices are diagonal.

The \emph{linear extension complexity} of a polytope~$P$ is the minimum number of facets of a (higher-dimensional) polytope which projects to~$P$.  The \emph{semidefinite (psd) extension complexity} of~$P$ is the minimum~$d$ such that an affine slice of the cone of $d$-by-$d$ positive semidefinite matrices projects to $P$.  These complexity measures can be captured in terms of the above notions of rank of a slack matrix:

\begin{theorem}[\cite{Yannakakis91,GouveiaParriloThomas12}]
\label{thm:yannakakis}
The linear extension complexity of a polytope $P$ is the nonnegative rank of a slack matrix of $P$.  The semidefinite (psd) extension complexity of~$P$ is the psd rank of a slack matrix of $P$.
\end{theorem}

A polytope may have different slack matrices associated with it, depending on which inner and outer description are used. By \thmref{thm:yannakakis} these slack matrices all have the same nonnegative and psd rank.

\subsection{Candidate matrix for lower bounding the correlation polytope}\label{ssec:candidatematrix}

One of our targets is the correlation polytope: $\text{COR}_n=\{xx^T : x \in \{0,1\}^n\}$. Fiorini et al.~\cite{FioriniMassarPokuttaTiwaryDewolf2012} showed that lower bounds on the linear/semidefinite extension complexity of the correlation polytope imply lower bounds on several other polytopes of interest, including the Traveling Salesman Polytope.
The next lemma from \cite{padberg89} gives a family of matrices that occur as a submatrix of the slack matrix of the correlation polytope.
\begin{lemma}
\label{lem:quad_poly}
Let $p(z)=a+bz+cz^2$ be a single-variate degree-$2$ polynomial that is nonnegative on nonnegative integers.  The matrix $M(x,y)=p(|x \wedge y|)$ for $(x,y) \in \{0,1\}^n$ is a submatrix of a slack matrix for the correlation polytope $\text{COR}_n$.
\end{lemma}

\begin{proof}
As $p$ is nonnegative on nonnegative integers, $-bz-cz^2 \le a$ is a valid inequality for integers $z \ge 0$.  Note that $\Tr(xx^Tyy^T)=|x\wedge y|^2$ and $\Tr(\diag(x)yy^T)=|x \wedge y|$ for all $x,y \in \{0,1\}^n$. Thus $\Tr((-b \cdot\diag(x)-c\cdot xx^T) yy^T) \le a$ is a valid inequality, whose slack is $p(|x \wedge y|)$. Note that the columns of $M$ are labeled by vertices of the correlation polytope $yy^T$ for $y \in \{0,1\}^n$ and likewise the constraints are labeled by $xx^T$ for $x \in \{0,1\}^n$.
\end{proof}

Later in this paper we will consider the matrix $M(x,y)=(|x \wedge y|-1)(|x \wedge y|-2)$ and its associated query problem $f(x)=(|x|-1)(|x|-2)$, where $|x|$ denotes the Hamming weight of the Boolean string $x$.

\subsection{Polynomials}
We will study two types of polynomials that are obviously nonnegative on the Boolean cube: nonnegative literal polynomials and sum-of-squares polynomials.

\begin{definition}[nonnegative literal degree]
A nonnegative literal polynomial is a nonnegative linear combination of products of variables and negations of variables, i.e., it can be written as
\[
p(x)=\sum_{S \subseteq [n]} \sum_{b \in \01^{|S|}} \alpha_{S,b} \prod_{i \in S} ((-1)^{b_i}x_i + b_i)
\]
where each $\alpha_{S,b} \ge 0$.  Its degree is $\max\{|S| : \alpha_{S,b}\neq 0\}$.
The nonnegative literal degree of $f:\01^n \rightarrow \R_+$, denoted $\ldeg_{+}(f)$, is the minimum degree of a nonnegative literal polynomial $p$ that equals $f$ on $\01^n$.
\end{definition}
Such~$p$ are also called \emph{nonnegative juntas}~\cite{clrs:csplp}. 

\begin{definition}[sum-of-squares degree]
Let $d$ be a natural number.
A sum-of-squares polynomial of degree~$d$ is a polynomial $p$ that can be written in the form
\[
p(x)=\sum_{i \in {\cal P}}p_i(x)^2,
\]
where $\cal P$ is a finite index set and the $p_i$ are polynomials of degree $\leq d$.
The sum-of-squares (sos) degree of $f: \01^n \rightarrow \R_+$, denoted $\deg_{sos}(f)$, is the minimum~$d$ for which such a~$p$ equals $f$ on $\01^n$. 
\end{definition}
Note that a sum-of-squares polynomial of degree~$d$ is actually a polynomial of degree $2d$; we allow this slight abuse of notation in order to give a clean characterization in Theorem~\ref{thm:polyn} below.

\subsection{The Sherali-Adams and Lasserre hierarchies}
\label{sec:hierarchies}
Consider the optimization problem
\begin{equation}
\label{eq:opt}\alpha(f) = \max_{x \in \01^n} f(x) \enspace 
\end{equation}
where $f$ is given by a multilinear polynomial.  Many important optimization problems can be cast in this framework, including NP-hard ones. For example finding the maximum cut in a graph $G=(V,E)$ with $n$ vertices corresponds to the 
quadratic function $f(x)=\sum_{(i,j) \in E} x_i (1-x_j)$.

If $c \ge \alpha(f)$, then the function $c-f$ is nonnegative on $\01^n$.  One way we can witness this is by expressing $c-f$ as a polynomial which is obviously nonnegative for all $x\in\01^n$.  The \emph{Sherali-Adams hierarchy}~\cite{SheraliAdams90} looks for a witness in the form of a nonnegative literal polynomial.  The sum-of-squares or \emph{Lasserre hierarchy} looks for a witness in the form of a sum-of-squares polynomial~\cite{Lasserre01, Parrilo00,Shor87}. 

If we can find a nonnegative literal polynomial $p$ of degree $d$ such that $c-f(x)=p(x)$, then this witnesses that the optimal value is upper bounded as $\alpha(f) \le c$.  Moreover, determining if the nonnegative literal polynomial degree of $c-f(x)$ is at most $d$ can be formulated as a linear program of size $n^{O(d)}$.  The value of the $d$-round Sherali-Adams relaxation for~\eqref{eq:opt} is the smallest value of $c$ such that $c-f(x)$ is a degree-$d$ nonnegative literal polynomial. Thus the smallest $d$ for which a Sherali-Adams relaxation certifies an \emph{optimal} upper bound is exactly the nonnegative literal degree $\ldeg_{+}(\alpha(f)-f)$ of the function $\alpha(f) - f$. 

Similarly, if we can find polynomials $p_i: \01^n \rightarrow \R$ of degree at most $d$, such that $c-f(x) = \sum_i p_i(x)^2$, then this witnesses that $\alpha(f) \le c$.  Moreover, searching for such polynomials $p_i$ can be expressed as a semidefinite program of size $n^{O(d)}$.  The smallest value of $c$ such that $c-f$ is degree-$d$ sum-of-squares is known to be equivalent to the relaxation of~\eqref{eq:opt} given by the $d^{th}$ level of the Lasserre hierarchy.  The level of the Lasserre hierarchy required to exactly capture~\eqref{eq:opt} is thus $\deg_{sos}(\alpha(f) - f)$.  

\section{Randomized query complexity in expectation}\label{sec:re}

In this section we study classical randomized query complexity in expectation, characterize it by the nonnegative literal degree, and relate it to the Sherali-Adams hierarchy.

\subsection{Definition}
\label{sec:randomized}
We define a randomized model of computing a function in expectation.  A \emph{randomized decision tree} is a probability distribution $\mu$ over deterministic decision trees.  We consider deterministic decision trees with leaves labeled by nonnegative real numbers.  A randomized decision tree computes a function $f: \01^n \rightarrow \R_+$ if for every $x \in \01^n$ the expected output of the tree on input $x$ is $f(x)$.  The \emph{cost} of such a tree is, as usual, the maximum cost, that is the length of a longest path from the root to a leaf, of a deterministic decision tree that has nonzero $\mu$-probability.

\begin{definition}
The randomized query complexity of computing $f$ in expectation,  denoted $\RE(f)$, is the minimum cost among all randomized decision trees that compute $f$ in expectation.
\end{definition}

\subsection{Characterization of $\RE(f)$ by polynomials}

We now show that $\RE(f)$ is characterized by the nonnegative literal degree.

\begin{theorem}\label{thRE=ldeg+}
Let $f: \01^n \rightarrow \R_+$.  Then $\RE(f)=\ldeg_{+}(f)$.
\end{theorem}

\begin{proof}
\underline{$\RE(f) \geq \ldeg_+(f)$.}  
We need to show how a randomized decision tree induces a nonnegative literal polynomial.
First consider a deterministic decision tree $T$ with leaves labeled by nonnegative real
numbers.  For each path $p$ from root to leaf, we construct a literal monomial $m_p$ where $x_i$ appears in $m_p$ if $x_i=1$ is on $p$, and $1-x_i$ appears if $x_i=0$ is on $p$.  The coefficient $\alpha_p$ of $m_p$ is the label of the leaf of $p$.  If we let $q_T(x)=\sum_{\text{paths } p} \alpha_p m_p(x)$ then we have that $q_{T}(x)$ is equal to the output of the tree on input $x$.  Moreover the degree of $q_{T}$ is at most the depth of $T$.  Now for a randomized decision tree that chooses a deterministic decision tree $T$ with probability $\mu(T)$, we set the polynomial $r(x) =\sum_T \mu(T) q_T(x)$, which gives a nonnegative literal representation of $f$.  

\underline{$\RE(f) \leq \ldeg_+(f)$.}  Let 
$$
p(x) =\sum_{S \subseteq [n]} \sum_{b \in \01^{|S|}} \alpha_{S,b} \prod_{i \in S} ((-1)^{b_i}x_i + b_i)
$$
be a nonnegative literal polynomial representing~$f$ of degree $\ldeg_+(f)$.  Let $M=\sum_{S,b} \alpha_{S,b}$.  The algorithm chooses $S,b$ with probability $\alpha_{S,b}/M$ and query all $i \in S$ to evaluate $a_{S,b}=\prod_{i \in S} ((-1)^{b_i}x_i + b_i)$.  Output~$M\cdot a_{S,b}$.  The expected output on input~$x$ equals~$p(x)$, and the number of queries is $\leq\ldeg_+(f)$.
\end{proof}

Referring back to Section~\ref{sec:hierarchies}, this gives a connection between randomized query complexity in expectation and the Sherali-Adams hierarchy: the smallest $d$ for which a Sherali-Adams relaxation certifies the optimal upper bound $\alpha(f)$ on the maximization problem~\eqref{eq:opt}, is exactly $\RE(\alpha(f) - f)$.

\section{Quantum query complexity in expectation}\label{sec:qe}

Here we study \emph{quantum} query complexity in expectation, characterize it by sum-of-squares degree, and relate it to Lasserre. We assume familiarity with quantum computing~\cite{nielsen&chuang:qc} and query complexity~\cite{buhrman&wolf:dectreesurvey}.

\subsection{Definition}

We define the quantum query complexity of computing a function $f:\01^n\to \R_+$ in expectation. A $T$-query algorithm is described by unitaries $U_0, \ldots, U_T$ and a POVM $\{E_\theta\}_{\theta \in \Theta}$, where each $E_\theta$ is a psd matrix labeled by nonnegative real~$\theta$, and $\sum_{\theta \in \Theta} E_\theta=I$.  As usual, on input $x$ the query algorithm proceeds from the initial state $\ket{\bar 0}$ by alternately applying a unitary and the query oracle $O_x$ (which maps $\ket{i,b}\mapsto\ket{i,b\oplus x_i}$), so that the state of the algorithm after $t$ queries is 
$\ket{\psi_x^t} = U_t O_x \ldots O_x U_1 O_x U_0\ket{\bar 0}.$
Let $E=\sum_{\theta \in \Theta} \theta E_\theta$.  As the probability of output $\theta$ upon measuring 
$\ket{\psi_x^T}$ is $\Tr(E_\theta \ketbra{\psi_x^T}{\psi_x^T})$, the expected value of the output is $\Tr(E \ketbra{\psi_x^T}{\psi_x^T})$.
The algorithm \emph{computes~$f$ in expectation} if $f(x)=\Tr(E \ketbra{\psi_x^T}{\psi_x^T})$ for every $x \in \01^n$.  

\begin{definition}
The quantum query complexity of computing $f$ in expectation, denoted $\QE(f)$, is the minimum $T$ for which there is a $T$-query quantum algorithm computing $f$ in expectation.
\end{definition}

\subsection{Characterization of $\QE(f)$ by polynomials}

We now adapt the polynomial method~\cite{BBCMW01} to characterize $\QE(f)$. The key is the following lemma, which says that the amplitudes of the final state of a $T$-query algorithm are degree-$T$ polynomials in~$x$:

\begin{lemma}[\cite{BBCMW01}]\label{lem:key_poly}
The state $\ket{\psi_x^t}$ of a quantum query algorithm on input $x$ after $t$ queries can be written as $\sum_{i,z} \alpha_{i,z}(x) \ket{i,z}$, where each $\alpha_{i,z}(x)$ is an $n$-variate multilinear polynomial in $x$ of degree~$\leq t$.
\end{lemma}


\begin{theorem}
\label{thm:polyn}
Let $f:\01^n\rightarrow\R_+$.
Then $\QE(f) = \deg_{sos}(f)$.
\end{theorem}

\begin{proof}
\underline{$\QE(f) \geq \deg_{sos}(f)$.}
Say there is a $T$-query algorithm to compute $f$ in expectation.  Then 
\[
f(x) = \sum_\theta \theta \bra{\psi_x^T} E_\theta \ket{\psi_x^T}.
\]
As the coefficients $\theta$ are nonnegative real numbers, it suffices to show that each term $\bra{\psi_x^T} E_\theta \ket{\psi_x^T}$ can be written as the sum of squares of polynomials of degree at most~$T$.  

Let $E_\theta = \sum_i \lambda_i \ketbra{e_\theta^i}{e_\theta^i}$ be the eigenvalue decomposition of $E_\theta$, 
where each $\lambda_i \ge 0$.  Then 
\[
\bra{\psi_x^T} E_\theta \ket{\psi_x^T} = \sum_i \lambda_i |\braket{\psi_x^T}{e_\theta^i}|^2 \enspace. 
\]
We have that $\braket{\psi_x^T}{e_\theta^i}$ is a linear combination of amplitudes of $\ket{\psi_x^T}$, hence by Lemma~\ref{lem:key_poly} it is a degree $\le T$ polynomial in $x$. Since $\lambda_i \ge 0$ this gives a representation of $\bra{\psi_x^T} E_\theta \ket{\psi_x^T}$ as a sum-of-squares polynomial of degree $\le T$.  Hence $T\geq \deg_{sos}(f)$.

\underline{$\QE(f) \leq \deg_{sos}(f)$.}
Let $d=\deg_{sos}(f)$.
We first exhibit a quantum algorithm for the special case where $f=p^2$ for some degree-$d$ polynomial~$p$.
This is inspired by the proof of~\cite[Theorem~2.3]{wolf:nqj}.
Let $p=\sum_s\widehat{p}(s)(-1)^{x\cdot s}$ be the Fourier representation of~$p$, where $s$ ranges over $\01^n$.
Because $p$ has degree~$d$, we have $\widehat{p}(s)\neq 0$ only if $|s|\leq d$.
The algorithm is as follows:
\begin{enumerate}
\item Prepare $n$-qubit state $c\sum_s\widehat{p}(s)\ket{s}$, where $c=1/\sqrt{\sum_s\widehat{p}(s)^2}$ is a normalizing constant.
\item Apply a unitary that maps $\ket{s}\mapsto(-1)^{x\cdot s}\ket{s}$ for all $s$ of weight $|s|\leq d$; one can show that this can be implemented using $d$ queries.
\item Apply the $n$-qubit Hadamard transform to the state.
\item Measure the state and output $2^n/c^2$ if the measurement result was $0^n$, otherwise output~0.
\end{enumerate}
Note that the amplitude of the basis state $\ket{0^n}$ after step~3 is
$$
\frac{c}{\sqrt{2^n}}\sum_s\widehat{p}(s)(-1)^{x\cdot s}=\frac{c}{\sqrt{2^n}}p(x).
$$
Hence the probability that the final measurement results in outcome $0^n$ is $(\frac{c}{\sqrt{2^n}}p(x))^2$, and the expected value of the output is $(\frac{c}{\sqrt{2^n}}p(x))^2\cdot 2^n/c^2=p(x)^2=f(x)$, as desired.

Now consider the general case where $f=\sum_{i\in{\cal P}} p_i^2$.
The algorithm chooses one $i\in{\cal P}$ uniformly at random and runs the above algorithm to produce an output with expected value $p_i(x)^2$.  It finally outputs that output multiplied by $|{\cal P}|$. Clearly, the algorithm uses 
at most $d$ queries to~$x$, and the expected value of its final output is 
$$
\frac{1}{|{\cal P}|}\sum_i p_i(x)^2|{\cal P}|=\sum_i p_i(x)^2=f(x).
$$
Hence $\QE(f) \leq d=\deg_{sos}(f)$.
\end{proof}

This gives a surprising connection between quantum query complexity in expectation and the Lasserre hierarchy: 
the smallest level $d$ of the Lasserre hierarchy that certifies the optimal upper bound $\alpha(f)$ on the 
maximization problem~\eqref{eq:opt}, is exactly $\QE(\alpha(f) - f)$.

\section{Gaps and relations between $\RE(f)$ and $\QE(f)$}\label{secqcgap}

For some $f:\01^n\rightarrow\R_+$, the quantum query complexity in expectation~$\QE(f)$ can be \emph{much} smaller than its classical counterpart~$\RE(f)$.  An extreme example is the $n$-bit function $f(x)=(|x|-1)^2$, where $\QE(f)=1$ by Theorem~\ref{thm:polyn}, but $\RE(f)=n$. The latter holds because on the all-0 input the algorithm needs to produce a nonzero output with positive probability, but on weight-1 inputs it can never output anything nonzero, hence a classical algorithm needs $n$ queries on the all-0 input. 

In contrast, if the range of $f$ is Boolean, then $\QE(f)$ is at most polynomially smaller than $\RE(f)$:

\begin{theorem}\label{thBooleanREQE}
For every $f:\01^n\rightarrow\01$ we have $\RE(f) \leq 16 \QE(f)^3$.
\end{theorem}

\begin{proof}
The result follows by chaining the following three inequalities:
\begin{enumerate}
\item $\RE(f)$ is obviously at most the deterministic decision tree complexity of $f$, denoted $D(f)$;
\item $D(f)\leq 2 \deg(f)^3$ by a result of Midrijanis~\cite[Theorem~4]{midrijanis:exact};
\item $\deg(f)\leq 2 \QE(f)$, because by Theorem~\ref{thm:polyn} a $T$-query $\QE$-algorithm gives a degree-$T$ sum-of-squares polynomial that represents $f$, which is a polynomial of degree $\leq 2T$.
\end{enumerate}
\vspace*{-2em}
\end{proof}

The main reason this query complexity result is interesting is that the analogous statement for \emph{communication} complexity is equivalent to the longstanding log-rank conjecture!  The communication version of Theorem~\ref{thBooleanREQE} would say that for all \emph{Boolean} matrices $M$, the quantum and classical communication complexity of computing $M$ in expectation are at most polynomially far apart.  As noted by Fiorini et al.~\cite{FioriniMassarPokuttaTiwaryDewolf2012}, this is equivalent to $\log \rk_+(M) \leq \polylog(\psdrk(M))$, which in turn is equivalent to the log-rank conjecture.  Presumably such a communication version will be substantially harder to prove than the above query version.  However, in many cases results in query complexity ``mirror'' (often much harder) results in communication complexity, so our Theorem~\ref{thBooleanREQE} may be viewed as (weak) evidence for the log-rank conjecture.

\section{Quantum query complexity lower bound}\label{ssecsoslowerbound}
In this section we show that the function $f(x)=(|x|-1)(|x|-2)$ has $\QE(f) = \Omega(\sqrt{n})$.  We do this 
by showing the corresponding lower bound on the sum-of-squares degree of $f$, adapting techniques
from approximation theory commonly used to show quantum query lower bounds in the bounded-error model.


We do this by using Theorem~\ref{thm:polyn} and bounding the sum-of-squares degree.  As is common in query complexity lower bounds by the polynomial method \cite{BBCMW01}, we will use a symmetrization argument to define a single-variate polynomial $Q:\R\rightarrow \R$ that behaves well on $[n]$, and then use Markov's lemma from approximation theory to bound the degree of $Q$.

A new complication in our setting is the following.  If $f(x) = \sum_i p_i(x)^2$ then we would like to define a ``symmetrized'' polynomial $g: [n] \rightarrow \R$ where $g(k)=\E_{x: |x|=k} \left[ \sum_i p_i(x)^2 \right]$.  We do not know how to show, however, that $g$ remains a nonnegative polynomial.  To get around this, we define symmetrized polynomials $q_i(k)=\E_{x: |x|=k} \left[p_i(x)\right]$ for each $p_i$ individually, then recombine the symmetrized polynomials as $Q(k) = \sum_i q_i(k)^2$.  We are then able to bound the sum-of-squares degree of~$Q$.

\begin{theorem}
Let $f(x)=(|x|-1)(|x|-2)$ for $x \in \01^n$.  Then 
$\displaystyle \deg_{sos}(f) \ge \sqrt{n/48}$.
\end{theorem}

\begin{proof}
Suppose that $f$ can be expressed as
\begin{equation*}
f(x) = \sum_{i} p_{i}(x)^{2},
\end{equation*}
where $\deg(p_{i})\leq T$ for all $i$. Let $q_{i} : [n] \to \mathbb{R}$ be defined as $q_{i}(k) = \mathbb{E}_{|x| = k} [p_{i}(x)]$.  By a standard symmetrization argument \cite{minskypapert87}, each $q_i$ is a polynomial of 
degree at most $T$.  Now consider
\begin{equation*}
Q(k) = \sum_{i} q_{i}(k)^{2},
\end{equation*}
which is a nonnegative polynomial in $k$ of degree at most $2T$. It satisfies $Q(0) = 2$, since there is only one $x$ of weight~0. Also, $Q(1) = Q(2) = 0$ since $f(x) = 0$ for $|x| \in \{1, 2\}$. The zeroes of a nonnegative polynomial must have even multiplicity, so at least $2$. Therefore there must exist a polynomial $q$ of degree at most $2T - 4$ such that
\begin{equation*}
Q(k) = (k - 1)^2 (k - 2)^2 q(k).
\end{equation*}
By convexity of the quadratic function, we find that
\begin{equation*}
Q(k) = \sum_{i} |q_{i}(k)|^{2} =  \sum_{i} \big| \mathbb{E}_{|x| = k} [p_{i}(x)] \big|^2 \leq \sum_{i} \mathbb{E}_{|x| = k} \big[ |p_{i}(x)|^{2} \big] = \mathbb{E}_{|x| = k} [f(x)] = (k - 1)(k - 2),
\end{equation*}
which implies
\begin{equation}
\label{eq:qbounded}
q(k) \leq 1/(k - 1)(k - 2).
\end{equation}
Note that $q(k)\leq 1/6$ for all integers $k\in\{4,\ldots,n\}$.\footnote{While we know that $q$ is nonnegative on $[n]$, we will not use this information.}  
We now simply lower bound the degree of $q$ using the following lemma of Markov:

\begin{lemma}[Markov]
If $q$ is a real polynomial then 
$\displaystyle\deg(q) \ge \sqrt{\frac{n}{2} \cdot \frac{\max_{x\in[0,n]} |q'(x)|}{\max_{x\in[0,n]} |q(x)|}}.$
\end{lemma}
Here $q'$ denotes the derivative of $q$.  Since $q(0)=Q(0)/4=1/2$, we know that the maximum value of $q$ in the interval $[0,n]$ is at least $1/2$.  Now suppose $\max_{x\in[0,n]} |q(x)|=c \ge 1/2$, and say that this maximum is attained at $x^*$. Since $q(k)\leq 1/6$ for all integers $k\in\{4,\ldots,n\}$, we know $x^*$ is at most distance $4$ from an~$x$ where $q(x) \leq 1/6$.  Thus $|q'(x)|\geq (6c-1)/24$ for some $x\in[0,n]$.  This, together with $c\geq 1/2$, shows that the ratio in Markov's lemma is at least 
$$
\frac{6c-1}{24c}=\frac{1}{4}-\frac{1}{24c} \ge \frac{1}{6}.
$$
Thus overall we obtain $2T\geq \deg(q) \ge \sqrt{\frac{n}{2}\frac{1}{6}}=\sqrt{n/12}$, implying the lower bound. 
\end{proof}

We note that stronger lower bounds on sum-of-squares degree are known for related functions.  Let $k=\floor{\tfrac{n}{2}}$ and consider $g(x)=(x_1 + \cdots + x_n -k)(x_1 + \cdots + x_n - k -1)$.  This polynomial is nonnegative on all $x \in \01^n$, and the induced matrix $M_g(x,y)=g(x \wedge y)$ is a submatrix of the slack matrix of the correlation polytope by \lemref{lem:quad_poly}.  For odd $n$, Grigoriev~\cite{grigoriev01} shows that the sum-of-squares degree of $g$ is $\floor{\tfrac{n}{2}}$ (see also \cite{laurent03}). Blekherman et al.~\cite{bgp:sosonhypercube} show that $g$ even has high \emph{rational} sum-of-squares degree: if the product $pg$ has sos degree $d$, where $p$ is an sos polynomials of degree~$r$, then $r+d \ge \floor{\tfrac{n}{2}}$.

Our lower bound technique is quite different from those used in these works, and is more closely related to works showing bounds on the minimum degree of a polynomial that \emph{approximates} a function in $\ell_\infty$ norm. In fact, our proof has recently been extended by Arunachalam, Yuen, and the last author~\cite{ayw:approxbound} to show that this $\Omega(\sqrt{n})$ sos-degree lower bound remains valid for functions $g$ that approximate $f$ pointwise up to additive error $O(1/n)$.  This is important because the very recent framework of Lee et al.~\cite{lrs:cspsdp} uses lower bounds on the sum-of-squares degree of a function that approximates $f$ pointwise 
to show lower bounds on the psd rank of a matrix associated with $f$.

\section{Psd rank and query complexity in expectation}
\label{sec:comm-query}

\subsection{Psd rank characterizes two-way quantum communication complexity}

Fiorini et al.~\cite{FioriniMassarPokuttaTiwaryDewolf2012} defined a \emph{one-way} model of quantum communication to compute a matrix in expectation, and showed that this complexity is characterized by the logarithm of the psd rank.  We show below that this characterization continues to hold for the more general \emph{two-way} communication model, which allows multiple rounds of communication between the two parties Alice and Bob.  Hence one-way and two-way quantum communication complexity are the same for computation in expectation. 

We will not formally define the model of two-way quantum communication complexity (see~\cite{wolf:qccsurvey} for more technical details), instead just highlighting the differences of the model of computing a function in expectation 
to the normal model.  As usual, Alice and Bob each start with their own input, $x$ and $y$ respectively, and then the protocol specifies whose turn it is to speak and what message they send to the other party. At the end of the protocol Bob must output a \emph{nonnegative} number, which is a random variable~$z$ that depends on the inputs $x$ and $y$ as well as on the internal randomness of the protocol.  

The major difference with the usual model
is the notion of when a protocol is correct.  Let $M$ be a matrix with nonnegative real entries whose rows are indexed by Alice's possible inputs, and whose columns are indexed by Bob's inputs. We say a protocol \emph{computes the matrix~$M$ in expectation} if, for every $(x,y)$, $M(x,y)$ equals the expected value of the output~$z$ on input~$(x,y)$.
As usual, the \emph{cost} of the protocol is the worst-case number of qubits that are communicated 
(summed over all rounds).

\begin{definition}
The quantum communication complexity of computing a matrix $M$ in expectation, denoted $\QCE(M)$, is the minimum $q$ such that there exists a quantum protocol of cost $q$ that computes $M$ in expectation. The minimum $q$ when we restrict to one-way protocols is denoted $\QCE^1(M)$.
\end{definition}

The following theorem shows that two-way quantum communication complexity is not more powerful than its 
one-way cousin: both are characterized by the psd rank.

\begin{theorem}\label{th:psdrkqce}
$\log \psdrk(M)\leq \QCE(M)\leq \QCE^1(f) \leq \log(\psdrk(M)+1)$.
\end{theorem}

\begin{proof}
The second inequality is obvious from the definitions.
Fiorini et al.~\cite{FioriniMassarPokuttaTiwaryDewolf2012} already showed how to construct a one-way protocol to compute $M$ in expectation using $\log(\psdrk(M)+1)$ many qubits, establishing the third inequality.  Thus we focus on the first inequality. Given a general protocol that computes~$M$ in expectation using $q$ qubits of communication, we need to construct a psd factorization of size~$2^q$.

The first step is to observe that one can replace the range of outputs of a multi-round $q$-qubit $\QCE$-protocol by $\{0,m\}$, where $m$ is the maximum output among all runs of the protocol: instead of outputting $m'$, just output $m$ with probability $m'/m$ and $0$ with probability $1-m'/m$, which preserves the expected value of the output.  In the remainder of the proof we assume for ease of notation that $m=1$.

Now use the Kremer-Yao lemma~\cite{kremer:thesis,yao:qcircuit} on this modified multi-round $q$-qubit communication protocol: its final state on input $x,y$ can be written as
$$
\sum_{i \in \01^{q+1}} \ket{a_i(x)} \ket{i_{q+1}}\ket{b_i(y)},
$$
where $\ket{a_i(x)}$ and $\ket{b_i(y)}$ are non-normalized states, and $i_{q+1}$ is the last bit of string $i$, 
corresponding to the output (0 or 1). Define $2^q$-by-$2^q$ psd matrices $A_x(i,j)=\braket{a_i(x)}{a_j(x)}$ where $i,j$ range over all $(q+1)$-bit strings that end in~1. Similarly define $B_y$. The expected value of the output is the probability to output~1:
$$
\norm{\sum_{i \in \01^q\times\{1\}}\ket{a_i(x)}\ket{1}\ket{b_i(y)}}^2=\sum_{i,j\in\01^q\times\{1\}}\braket{a_i(x)}{a_j(x)}\cdot\braket{b_i(y)}{b_j(y)}=\Tr(A_x B_y).
$$  
Thus a multi-round $q$-qubit protocol gives a psd factorization of~$M$ of size~$2^q$.
\end{proof}

\subsection{Upper bounds on psd rank from quantum algorithms}

We now show that efficient quantum query algorithms for computing functions $f:\01^n\to\R_+$ in expectation give rise to an efficient quantum communication protocol to compute the matrix $M_f(x,y)=f(x \wedge y)$ in expectation, and hence to a low-rank psd factorization of $M_f$. 
We state it more generally:

\begin{theorem}
\label{thm:comm-query}
Let $Y$ be a finite set. For every $y\in Y$, let $f_y:\01^n \rightarrow \R_+$ satisfy $\QE(f_y)\leq T$. Define a $2^n\times|Y|$ matrix $M$ by $M(x,y)=f_y(x)$.
Then $\QCE(M)\leq 2T(\log(n)+1)$, and hence $\psdrk(M)\leq (2n)^{2T}$.
\end{theorem}

\begin{proof}
The proof is very similar to an analogous statement by Buhrman, Cleve, and Wigderson~\cite{BuhrmanCleveWigderson98} for regular quantum communication complexity. Bob (who has input~$y$) runs a $T$-query algorithm for $f_y$; whenever he needs to make a query to~$x$ he sends the $(\log(n)+1)$-qubit query register to Alice, who applies the query and sends it back.  Thus every query is implemented using $2(\log(n)+1)$ qubits of communication, and the expected value of Bob's output is $f_y(x)$. The bound on the psd rank follows from Theorem~\ref{th:psdrkqce}.
\end{proof}

Lee et al.~\cite{lrs:cspsdp} independently proved a similar upper bound on psd rank, stated in terms of the sos degree of the $f_y$ rather than quantum query complexity (which are equal by Theorem~\ref{thm:polyn}).

The $\log n$ factor in Theorem~\ref{thm:comm-query} is necessary.  Consider the function $f(x)=(|x|-1)^2$.  Then $\QE(f)=1$ by Theorem~\ref{thm:polyn}.  On the other hand $\psdrk(M_f)\geq n/\sqrt{2}$: it is easy to see that the rank of~$M$ is at most the square of its psd rank, and the rank of $M_f(x,y)=(|x \wedge y|-1)^2$ is $n^2/2+1$ using~\cite[Section~4.1]{BW01}.

\subsection{Application: approximating the slack matrix of the matching polytope}\label{ssec:matchingpolytope}

Here we give an application of the above connection between query algorithms and psd rank, by deriving an exponentially-close entrywise approximation of the slack matrix~$S$ of the perfect matching polytope, by a matrix with psd rank not much bigger than $2^{\sqrt{n}}$. This shows a big difference to the case of nonnegative rank: Braun and Pokutta~\cite{braun&pokutta:matchrelax} show that any $\tilde{S}$ that is $O(1/n)$-close to $S$ needs nonnegative rank $2^{\Omega(n)}$.

Edmonds gave a complete description of the facets of the perfect matching polytope for the complete $n$-vertex graph $K_n$~\cite{edmonds65}.  The key are the \emph{odd-set} inequalities:  for a perfect matching $M$, viewed as a vector $M\in \01^{{n\choose 2}}$ of weight $m=n/2$, and an odd-sized set $U \subseteq [n]$, the associated inequality says $|\delta(U) \cap M| \ge 1$, where $\delta(U)\in\01^{{n\choose 2}}$ denotes the cut induced by~$U$.  In addition, there are $O(n^2)$ degree and nonnegativity constraints.
Thus the corresponding slack matrix~$S$ has columns indexed by all perfect matchings $M$ in $K_n$ and rows indexed 
by odd-sized sets $U$ with entries $S_{UM}=|\delta(U)\cap M|-1$.  There are $O(n^2)$ additional rows for the degree and nonnegativity constraints.

In Theorem~\ref{thtayloredgrover} in the appendix, we show that the $m$-bit function $g(z)=|z|-1$ can be approximated (in expectation) up to exponentially small error with quantum query complexity $O(m^{1/2+\eps}\log m)$. Define $f_M(x)=g(x_M)$, where $x_M$ denotes the restriction of $n$-bit string~$x$ to the $m$ positions in the support of~$M$.  Applying Theorem~\ref{thm:comm-query} and adding $O(n^2)$ rows to account for the other constraints gives:

\begin{theorem}
For every $\eps>0$ there exists a matrix $\tilde{S}$ of psd rank 
$2^{O(n^{1/2+\eps}(\log n)^2)}$ such that
\begin{enumerate} 
\item $S_{UM}-2^{-(n/2)^{2\eps}}\leq\tilde{S}_{UM}\leq S_{UM}$ for the $UM$-entries where $|\delta(U)\cap M|>(n/2)^{2\eps}$;
\item $\tilde{S}_{xy}=S_{xy}$ for all other entries.
\end{enumerate} 
\end{theorem}

\paragraph{Acknowledgments.}
We thank Srinivasan Arunachalam, David Steurer, Mario Szegedy and Henry Yuen for useful discussions, Sebastian Pokutta for useful discussions and for pointing us to~\cite{braun&pokutta:matchrelax}, and James Lee for sending us a version of~\cite{lrs:cspsdp}.

\bibliographystyle{alpha}

\newcommand{\etalchar}[1]{$^{#1}$}

\appendix
\section{A tailored quantum search algorithm}

The \emph{search problem} is the following: we have an $m$-bit input $z$ that we can access by means of queries, and our goal is to find an index $i\in[m]$ such that $z_i=1$. Such an $i$ will be called a ``solution''.  
The number of solutions is the Hamming weight of the input, denoted~$|z|$. 
Grover's algorithm~\cite{grover:search,bhmt:countingj} solves this problem using $O(\sqrt{m})$ queries.  We will use the following two variants:
\begin{itemize}
\item There is a quantum algorithm using $O(\sqrt{m/t})$ queries that finds a solution with probability at least 1/2 if $|z|\in[t,2t]$. 
\item There is a quantum algorithm using $O(\sqrt{m/t})$ queries that finds a solution with certainty if $|z|=t$. 
\end{itemize}
We combine these variants of Grover to prove the following theorem, similar to~\cite[Theorem~3]{bcwz:qerror}:  

\begin{theorem}\label{thtayloredgrover}
For every integer $\ell>0$ there exists a quantum algorithm that makes $O(\sqrt{m\ell}\log m)$ queries to input $z\in\01^m$ and that has the following properties:
\begin{enumerate} 
\item If $z=0^m$ then the algorithm outputs ``no solution'' with certainty.
\item If $|z|\in\{1,\ldots,\ell\}$ then the algorithm outputs a solution with certainty.
\item If $|z|>\ell$ then the algorithm outputs a solution with probability $\geq 1-2^{-\sqrt{\ell|z|}}$.
\end{enumerate}
\end{theorem}

\begin{proof}
The algorithm is as follows:
\begin{enumerate}
\item Run exact Grover $\ell$ times, once for each of the possibilities $t=1,2,\ldots,\ell$.
\item For $i=\floor{\log\ell},\ldots,\floor{\log m}$:
Run $\ceil{\sqrt{\ell 2^{i+1}}}$ times the version of Grover that assumes $|z|\in [2^i,2^{i+1}]$.
\item Check each of the indices produced by these runs (using one query per index).
\item Output a solution if one was found, and output ``no solution'' otherwise.
\end{enumerate}
Clearly, the algorithm behaves as promised if $|z|\leq\ell$. Now suppose $|z|>\ell$ and let $i$ be the unique integer such that $|z|\in[2^i,2^{i+1})$.  For that~$i$, each of the $\ceil{\ell\sqrt{2^{i+1}}}$ runs of Grover has probability $\geq 1/2$ of producing a solution, hence the probability of \emph{not} finding a solution is $\leq 2^{-\ceil{\sqrt{\ell 2^{i+1}}}}\leq 2^{-\sqrt{\ell|z|}}$ in this case.

It remains to bound the query complexity of the algorithm.
The number of queries used in step~1 is
$$
\sum_{t=1}^\ell O(\sqrt{m/t})=O(\sqrt{m\ell}).
$$
The number of queries used in step~2 is
$$
\sum_{i=\floor{\log\ell}}^{\floor{\log m}} \ceil{\sqrt{\ell 2^{i+1}}}O\left(\sqrt{m/2^i}\right)=O(\sqrt{m\ell}\log m).
$$
The total number of runs of (versions of) Grover's algorithm is 
$$
\ell+\sum_{i=\floor{\log\ell}}^{\floor{\log m}}\ceil{\sqrt{\ell 2^i}}=O(\sqrt{m\ell}).
$$
Since each such run produces one index that needs to be checked, the number of queries made in step~3 is $O(\sqrt{m\ell})$. Thus the overall query complexity is $O(\sqrt{m\ell}\log m)$ as promised.
\end{proof}

We can derive from this a function $f:\01^m\to\R_+$ that approximates $|z|-1$ extremely well, and that has quantum query complexity in expectation not much bigger than $\sqrt{m}$:

\begin{theorem}
For every $\eps>0$ there exists a function $f:\01^m\to\R_+$ satisfying $\QE(f)=O(m^{1/2+\eps}\log m)$ and 
\begin{enumerate} 
\item $f(0^m)=0$.
\item If $|z|\in\{1,\ldots,\ell\}$ then $f(z)=|z|-1$.
\item If $|z|>\ell$ then $|z|-1-2^{-m^{2\eps}}\leq f(z)\leq|z|-1$.
\end{enumerate}
\end{theorem}

\begin{proof}
Set $\ell=m^{2\eps}$.  Run the algorithm of Theorem~\ref{thtayloredgrover}, which uses $O(m^{1/2+\eps}\log m)$ queries. If $|z|\geq 1$, it finds a solution (i.e., an $i\in[m]$ such that $z_i=1$) with very high probability. If it did not find a solution the algorithms outputs~0.  If, on the other hand, $i$ is a solution then the algorithm queries a uniformly random index $j\neq i$ and outputs $z_j\cdot(m-1)$. Let $f(z)$ be the expected output of this algorithm on input $z$.

If $z=0^m$ the algorithm always outputs 0, establishing the first property. If $|z|\in\{1,\ldots,\ell\}$ then $i$ is a solution with certainty, and the expected value of the output is $\Pr[z_j=1]\cdot(m-1)=\frac{|z|-1}{m-1}\cdot(m-1)=|z|-1$, establishing the second property. If $|z|>\ell$ then the algorithm finds a solution except with probability $2^{-m^{2\eps}}$, which implies the third property.
\end{proof}

\end{document}